\documentclass[letterpaper,12pt]{amsart}
\usepackage{amsmath,amsfonts,amsthm,mathrsfs,amssymb,graphicx,subfigure,bbm}
\usepackage[margin=1in]{geometry}

\usepackage{color}

\definecolor{myblue}{rgb}{0.3, 0.0, 0.85}
\definecolor{myviolet}{rgb}{0.5, 0.0, 0.5}

\usepackage{hyperref}

\hypersetup{
	colorlinks=true,
	linkcolor=myblue,
	citecolor=myviolet
}

\theoremstyle{plain}

\newtheorem{lem}{Lemma}[section]

\newtheorem{defn}{Definition}[section]

\newtheorem*{unlem}{Lemma}
\newtheorem*{unprop}{Proposition}

\newtheorem{snthm}{Theorem}
\newtheorem{snprop}{Proposition}

\title{Structure of Noncommutative Solitons: Existence and Spectral Theory}
\author{August J. Krueger}
\address{Amado Building, Technion Math.\ Dept., Haifa, 32000, Israel}
\email{ajkrueger@tx.technion.ac.il}
\author{Avy Soffer}
\address{110 Frelinghuysen Rd., Rutgers Math.\ Dept., Piscataway, NJ 08854, USA}
\email{soffer@math.rutgers.edu}

\begin{document}

\maketitle


\begin{abstract}
We consider the Schr\"odinger equation with a Hamiltonian given by a second order difference operator with nonconstant growing coefficients, on the half one dimensional lattice. This operator appeared first naturally in the construction and dynamics of noncommutative solitons in the context of noncommutative field theory. We construct a ground state soliton for this equation and analyze its properties. In particular we arrive at $\ell^{\infty}$ and $\ell^{1}$ estimates as well as a quasi-exponential spatial decay rate.\\
\ \\
\noindent \textsc{Mathematics subject classification.} 35Q40, 35Q55, 39A05.\\
\ \\
\noindent \textsc{Keywords.} Noncommutative soliton, spectral theory, NLS, DNLS.
\end{abstract}


\section{Introduction and Background}

The notion of noncommutative soliton arises when one considers the nonlinear Klein-Gordon equation (NLKG) for a field which is dependent on, for example, two ``noncommutative coordinates'', $x,y$, whose coordinate functions satisfy canonical commutation relations (CCR) $[X,Y]=i\epsilon$. This follows through the method of deformation quantization, see e.g. \cite{BaFl1} for a review and \cite{BaFl2} for applications. By going to a representation of the above canonical commutation relation, one can reduce the dynamics of the problem to an equation for the coefficients of an expansion in the Hilbert space representation of the above CCR, see e.g. \cite{DJN 1}\cite{DJN 2}\cite{GMS}. By restricting to rotationally symmetric functions the nocommutative deformation of the Laplacian reduces to a second order finite difference operator, which is symmetric, and with variable coefficient growing like the lattice coordinate, at infinity. Therefore, this operator is unbounded, and in fact has continuous spectrum $[0,\infty)$. These preliminary analytical results, as well as additional numerical results, were obtained by Chen, Fr\"ohlich, and Walcher \cite{CFW}. The dynamics and scattering of the (perturbed) soliton can then be inferred from the NLKG with such a discrete operator as the linear part. We will be interested  in studying the dynamics of discrete NLKG and discrete NLS equations with these hamiltonians.


We will be working with a discrete Schr\"odinger operator $L_0$ which can be considered either a discretization or a noncommutative deformation of the radial 2D negative Laplacian, $-\Delta^{\mathrm{2D}}_\mathrm{r} = -r^{-1}\partial_r r \partial_r$. We will briefly review both perspectives.

In 1D one may find a discrete Laplacian via
\begin{align*}
& x \in \mathbb{R} \  \xrightarrow{\ \mathrm{discrete} \ } \  n \in \mathbb{Z}, \quad -\Delta^{\mathrm{1D}} = -\partial^{2}_x \  \xrightarrow{\ \mathrm{discrete} \ } \  -D_{+}D_{-},
\end{align*}
where $D_{+}v(n) = v(n+1) - v(n), D_{-}v(n) = v(n) - v(n-1)$ are respectively the forward and backward finite difference operators. It is important to implement this particular combination of these finite difference operators due in order to ensure that the resulting discrete Laplacian is symmetric. In 2D one may find a discrete Laplacian via
\begin{align*}
&r = (x^{2} + y^{2})^{1/2} =  2 \rho^{1/2},\quad \rho \in \mathbb{R}_+ \  \xrightarrow{\ \mathrm{discrete} \ } \  n \in \mathbb{Z}_+, \\
&-\Delta^{\mathrm{2D}}_{\mathrm{r}} = -r^{-1}\partial_{r}r\partial_{r} = -\partial_\rho \rho \partial_\rho \quad \xrightarrow{\ \mathrm{discrete} \ } \quad -D_{+}MD_{-} = L_{0},
\end{align*}
where $Mv(n) = nv(n)$. For any 1D  continuous coordinate $x$ one may discretize a pointwise multiplication straightforwardly via $v^p(x) \  \xrightarrow{\ \mathrm{discrete} \ } \  v^p(n)$, where $n$ is a discrete coordinate.

One may also follow the so-called noncommutative space perspective. Here one considers the formal ''Moyal star deformation'' of the algebra of functions on $\mathbb{R}^{2}$:
\begin{align*}
\Phi_{1}\cdot\Phi_{2}(x,y) &= \Phi_{1}(x,y)\Phi_{2}(x,y) \\
\xrightarrow{\ \epsilon > 0\ } \quad \Phi_{1} \star \Phi_{2}(x, y) &= \exp[i(\epsilon / 2)(\partial_{x_1}\partial_{y_2} - \partial_{y_1}\partial_{x_2})] \Phi_{1}(x_1,y_1)\Phi_{2}(x_2,y_2)\lfloor_{ (x_j,y_j) = (x,y) }.
\end{align*}
One calls the coordinates, $x,y$, noncommutative in this context because the coordinate functions $X(x,y) = x$, $Y(x,y) = y$ satisfy a nontrivial commutation relation $X\star Y - Y\star X \equiv [X,Y] = i\epsilon$. This prescription can be considered equivalent to the multiplication of functions of $q, p$ in quantum mechanics where operator ordering ambiguities are set by the normal ordering prescription for each product. For $\Phi$ a deformed function of $r = (x^2 + y^2)^{1/2}$ alone: $\Phi = \sum_{n=0}^\infty v(n)\Phi_n$ where $v(n) \in \mathbb{C}$ and the $\{\Phi_{n}\}_{n=0}^\infty$ are distinguished functions of $r$: the projectors onto the eigenfunctions of the noncommutative space variant of quantum simple harmonic oscillator system. One may find for $\Phi$ a function of $r$ alone:
\begin{align*}
-\Delta^{\mathrm{2D}} \Phi &= -\Delta^{\mathrm{2D}}_{\mathrm{r}}\Phi = -r^{-1}\partial_{r}r\partial_{r}\Phi \\
\xrightarrow{\ \epsilon > 0\ }\quad \frac{2}{\epsilon} L_0\Phi_n &= \frac{2}{\epsilon} \left\{ \begin{array}{cc}
		- (n+1)\Phi_{n+1} + (2n + 1)\Phi_{n} - n \Phi_{n-1} &,\quad n > 0 \\
		- \Phi_{1} + \Phi_{0} &,\quad n = 0 .
	\end{array} \right.
\end{align*}
which may be transferred to $\frac{2}{\epsilon} L_0v(n)$, an equivalent action on the $v(n)$, due to the symmetry of $L_0$. Since the $\Phi_n$ are noncommutative space representations of projection operators on a standard quantum mechanical Hilbert space, they diagonalize the Moyal star product: $\Phi_m\star\Phi_n = \delta_{m,n}\Phi_n$. This property is shared by all noncommutative space representations of projection operators. Thereby products of the $\Phi_n$ may be transferred to those of the expansion coefficients: $v(n)v(n) = v^2(n)$.

See  B. Durhuus, T. Jonsson, and R. Nest \cite{DJN 1,DJN 2} (2001) and T. Chen, J. Fr\"ohlich, and J. Walcher \cite{CFW} (2003) for reviews of the two approaches. In the following we will work on a lattice explicitly so $x \in \mathbb{Z}_+$ will be a discrete spatial coordinate.

The principle of replacing the usual space with a noncommutative space (or space-time) has found extensive use for model building in physics and in particular for allowing easier construction of localized solutions, see e.g. \cite{fuzzy physics}\cite{NC soliton survey} for surveys. An example of the usefulness of this approach is that it may provide a robust procedure for circumventing classical nonexistence theorems for solitons, e.g. that of Derrick \cite{Derrick}, in the following sense. The act of deforming spatial coordinates to be noncommuting can be considered as effectively replacing a continuum problem with an effective lattice problem in which simple nonlinearities hold the same form. Lattice systems are often seen to exhibit a wider variety of solitary wave phenomena, e.g. breathers, and possibly enhanced stability. For example, stable solitons are prohibited in the NLKG one time and two commutative spatial coordinates but numerical evidence points to a possible meta-stability thereof for the case with noncommutative spatial coordinates conjectured in \cite{CFW}.

The NLKG variant of the equation we study here first appeared in the context of string theory and associated effective actions in the presence of background D-brane configurations, see e.g. \cite{GMS}. We have decided to look in a completely different direction. The NLS variant and its solitons can in principle be materialized experimentally with optical devices, suitably etched, see e.g. \cite{Segev review}. Thus the dynamics of NLS with such solitons may offer new and potentially useful coherent states for optical devices. Furthermore, we believe the NLS solitons to have special properties, in particular asymptotic stability as opposed to the asymptotic metastability of the NLKG solitons conjectured in \cite{CFW}.


We will be following a procedure for the proof of asymptotic stability which has become standard within the study of nonlinear PDE \cite{Avy NLS}. Crucial aspects of the theory and associated results were established by Buslaev and Perelman \cite{important results 1}, Buslaev and Sulem \cite{important results 2}, and Gang and Sigal \cite{important results 3}. Important elements of these methods are the dispersive estimates. Various such estimates have been found in the context of 1D lattice systems, for example see the work of A.I. Komech, E.A. Kopylova, and M. Kunze \cite{important results 4} and of I. Egorova, E. Kopylova, G. Teschl \cite{1D lattice decay estimates}, as well as the continuum 2D problem to which our system bears many resemblances, see e.g. the work of E. A. Kopylova and A.I. Komech \cite{2D}. Extensive results have been found on the asymptotic stability on solitons of 1D nonlinear lattice Schr\"odinger equations by F. Palmero et al. \cite{important results 5}, P.G. Kevrekidis, D.E. Pelinovsky, and A. Stefanov \cite{important results 6}, as well as S. Cuccagna and M. Tarulli \cite{CucTar}. Typically the literature on 1D lattice NLS systems focuses on cases where the free linear Schr\"odinger operator is given by the negative of the standard 1D discrete Laplacian. Our work is on a different free linear Schr\"odinger operator, $L_0$ defined above, which has some distinguishing properties. Important aspects of the application of these models to optical nonlinear waveguide arrays has been established by H.S. Eisenberg et al. \cite{important results 7}.


This work is the first of a series of papers (\cite{paper 01}, this one, and \cite{paper 03}) devoted to the construction, scattering, and asymptotic stability of radial noncommutative solitons with two noncommuting spatial coordinates. We have chosen to restrict our study to these solutions for a number of reasons: it builds upon the observations and results of \cite{CFW}; the radial cases allow one to work with effective 1D lattices and thereby standard Jacobi operators; for two noncommuting spatial coordinates the free radial system is equivalent to a known Jacobi operator spectral problem; the method proposed is by far the most illustrative for the given restrictions. The three papers are devoted to separate aspects of the problem in order of necessity. The organization of this work is as follows.

In \cite{paper 01} we focus on a key estimate that is needed for scattering and stability, namely the decay in time of solutions of relevant Schr\"odinger operators. Fortunately, for boundary perturbed operators, we find it is integrable, given by $t^{-1}\log^{-2}t$. The proof of this result is rather direct, and employs the generating functions of the corresponding generalized eigenfunctions, to explicitly represent and estimate the resolvent of the hamiltonian at all energies. We follow the general approach established by Jensen and Kato \cite{JenKat} and extended by Murata \cite{Murata} whereby time decay follows largely from the behavior of the resolvent near the threshold. From this one can see that for the chosen weights the estimate we find is optimal and should be optimal in general due to the elimination of the threshold resonance by boundary perturbations, by the generality of the method. We also conclude the absence of positive eigenvalues and singular continuous spectrum.

Previous results for the scattering theory of the associated noncommutative waves and solitons were found by Durhuus and Gayral \cite{noncommutative scattering}. In particular they find local decay estimates for the associated noncommutative NLS. They consider general noncommutative estimates for all for all even dimensions of pairwise noncommuting spaces. We consider radial solutions on 2D noncommutative space by alternative methods and find local decay for both the free Schr\"odinger operator as well as a class of rank one perturbations thereof. Our decay estimates are an improvement on those of \cite{noncommutative scattering} for this restricted class of solutions. An important element of this analysis is the study of the spectral properties of the free and boundary-perturbed Schr\"odinger operator. The boundary-perturbation is crucial to the work as it not only eliminates the threshold resonance of the free operator (thereby improving the time decay) as well as allows one to approximate and control solitons that are large only at the boundary via linear operators. We extend the linear analysis of Chen, Fr\"ohlich, and Walcher \cite{CFW} and reproduce some of their results with alternative techniques.

In this paper we address the construction and properties of a family of ground state solitons. These stationary states satisfy a nonlinear eigenvalue equation, are positive, monotonically decaying and sharply peaked for large spectral parameter. The proof of this result follows directly from our spectral results in this paper by iteration for small data and root finding for large data. The existence and many properties of solutions for a similar nonlinear eigenvalue equation were found by Durhuus, Jonssen, and Nest \cite{DJN 1}\cite{DJN 2}. We utilize a simple power law nonlinearity for which their existence proofs do not apply. We additionally find estimates for the peak height, spatial decay rate, norm bounds, and parameter dependence.

In \cite{paper 03} we focus on deriving a decay rate estimate for the Hamiltonian which results from linearizing the original NLS around the soliton constructed in this paper. We determine the full spectrum of this operator, which is the union of a multiplicity 2 null eigenvalue and a real absolutely continuous spectrum. This establishes a well-defined set of modulation equations \cite{Avy NLS} and points toward the asymptotic stability of the soliton.

In the conclusion of \cite{paper 03} we describe how the results can be applied to prove stability of the soliton we constructed in this paper. The issue of asymptotic stability of NLS solitons has been sufficiently well-studied in such a broad context that the proof thereof is often considered as following straightforwardly from the appropriate spectral and decay estimates, of the kind found in \cite{paper 03}. We sketch how the theory of modulation equations established by Soffer and Weinstein \cite{Avy NLS} can be used to prove asymptotic stability. Chen, Fr\"ohlich, and Walcher \cite{CFW} conjectured that in the NLKG case the corresponding solitons are unstable but with exponential long decay: the so-called metastability property, see \cite{Avy NLKG} . There is a great deal of evidence to suggest that this is in fact the case but a proof has yet to be provided. This will be the subject of future work.

\section{Notation}

Let $\mathbb{Z}_+$ and $\mathbb{R}_+$ respectively be the nonnegative integers and nonnegative reals and $\mathscr{H} = \ell^2(\mathbb{Z}_+,\mathbb{C})$ the Hilbert space of square integrable complex functions, e.g. $v: \mathbb{Z}_+ \ni x \mapsto v(x) \in \mathbb{C}$, on the 1D half-lattice with inner product $( \cdot , \cdot )$, which is conjugate-linear in the first argument and linear in the second argument, and the associated norm $||\cdot||$, where $||v|| = (v,v)^{1/2}$, $\forall v\in\mathscr{H}$. Where the distinction is clear from context $||\cdot|| \equiv ||\cdot||_{\mathrm{op}}$ will also represent the norm for operators on $\mathscr{H}$ given by $||A||_{\mathrm{op}} = \sup_{v \in \mathscr{H}}||v||^{-1}||Av||$, for all bounded $A$ on $\mathscr{H}$. Denote the lattice $\ell^1$ norm by $||\cdot||_1$ where $||v||_1 = \sum_{x=0}^\infty|v(x)|$, $\forall v \in \ell^1(\mathbb{Z}_+,\mathbb{C})$.

We denote by $\otimes$ the tensor product and by $z \mapsto \overline{z}$ complex conjugation for all $z \in \mathbb{C}$. We write $\mathscr{H}^*$ for the space of linear functionals on $\mathscr{H}$: the dual space of $\mathscr{H}$. For every $v \in \mathscr{H}$ one has that $v^* \in \mathscr{H}^*$ is its dual satisfying $v^{*}(w) = (v,w)$ for all $v,w \in \mathscr{H}$.  For every operator $A$ on $\mathscr{H}$ we take $\mathcal{D}(A)$ as standing for the domain of $A$. For each operator $A$ on $\mathscr{H}$ define $A^*$ on $\mathscr{H}^*$ to be its dual and $A^\dag$ on $\mathscr{H}$ its adjoint such that $v^*(Aw) = A^*v^*(w) = (A^\dag v, w)$ for all $v \in \mathcal{D}(A^\dag)$ and all $w \in \mathcal{D}(A)$. Let $\{\chi_{x}\}_{x=0}^{\infty}$ be the orthonormal set of vectors such that $\chi_{x}(x) = 1$ and $\chi_{x_{1}}(x_{2}) = 0$ for all $x_{2} \ne x_{1}$. We write $P_{x} = \chi_{x} \otimes \chi^{*}_{x}$ for the orthogonal projection onto the space spanned by $\chi_{x}$.

We define $\mathscr{T}$ to be the topological vector space of all complex sequences on $\mathbb{Z}_+$ endowed with topology of pointwise convergence, $\mathcal{B}(\mathscr{H})$ to be the space of bounded linear operators on $\mathscr{H}$, and $\mathcal{L}(\mathscr{T})$ to be the space of linear operators on $\mathscr{T}$, endowed with the pointwise topology induced by that of $\mathscr{T}$. When an operator $A$ on $\mathscr{H}$ can be given by an explicit formula through $A(x_{1},x_{2}) = (\chi_{x_{1}},A\chi_{x_{2}}) < \infty$ for all $x_{1},x_{2} \in \mathbb{Z}_{+}$ one may make the natural inclusion of $A$ into $\mathcal{L}(\mathscr{T})$, the image of which will also be denoted by $A$. We consider $\mathscr{T}$ to be endowed with pointwise multiplication, i.e. the product $uv$ is specified by $(uv)(x)=u(x)v(x)$ for all $u,v \in \mathscr{T}$.

We represent the \emph{spectrum} of each $A$ on $\mathscr{H}$ by $\sigma(A)$. We term each element $\lambda \in \sigma(A)$ a \emph{spectral value}. We write $\sigma_{\mathrm{d}}(A)$ for the \emph{discrete spectrum}, $\sigma_{\mathrm{e}}(A)$ for the \emph{essential spectrum}, $\sigma_{\mathrm{p}}(A)$ for the \emph{point spectrum}, $\sigma_{\mathrm{ac}}(A)$ for the \emph{absolutely continuous spectrum}, and $\sigma_{\mathrm{sc}}(A)$ for the \emph{singularly continuous spectrum}. Should an operator $A$ satisfy the spectral theorem there exist scalar measures $\{\mu_{k}\}_{k=1}^{n}$ on $\sigma(A)$ which furnish the associated spectral representation of $\mathscr{H}$ for $A$ such that the action of $A$ is given by multiplication by $\lambda \in \sigma(A)$ on $\oplus_{k=1}^{n}L^{2}(\sigma(A),\mathrm{d}\mu_{k})$. If $\mathscr{H} = \oplus_{k=1}^{n}L^{2}(\sigma(A),\mathrm{d}\mu_{k})$ we term $n$ the \emph{generalized multiplicity} of $A$. For an operator of arbitrary generalized multiplicity we will write $\mu^{A}$ for the associated operator valued measure, such that $A = \int_{\sigma(A)} \lambda\ \mathrm{d}\mu^{A}_{\lambda}$. For each operator $A$ that satisfies the spectral theorem, its spectral (Riesz) projections will be written as $P^{A}_{\mathrm{d}}$ and the like for each of the distinguished subsets of the spectral decomposition of $A$. Define $R^A_\cdot: \rho(A) \to \mathcal{B}(\mathscr{H})$, the resolvent of $A$, to be specified by $R^A_z := (A - z)^{-1}$, where $\rho(A) := \mathbb{C} \setminus \sigma(A)$ is the resolvent set of $A$ and where by abuse of notation $zI \equiv z \in \mathcal{B}(\mathscr{H})$ here.

Allow an \emph{eigenvector} of $A$ to be a vector $v \in \mathscr{H}$ for which $Av = \lambda v$ for some $\lambda \in \mathbb{C}$. Should $A$ admit inclusion into $\mathcal{L}(\mathscr{T})$, we define a \emph{generalized eigenvector} of $A$ be a vector $\phi \in \mathscr{T} \setminus \mathscr{H}$ which satisfies $A\phi = \lambda \phi$ for some $\lambda \in \mathbb{C}$ such that $\phi(x)$ is polynomially bounded, which is to say that there exists a $p \ge 0$ such that $\lim_{x \nearrow \infty}(x+1)^{-p}\phi(x) = 0$. We define a \emph{spectral vector} of $A$ to be a vector which is either an eigenvector or generalized eigenvector of $A$. We define the subspace of spectral vectors associated to the set $\Sigma \subseteq \sigma(A)$ to be the \emph{spectral space over $\Sigma$}.

We write $\partial_z \equiv \frac{\partial}{\partial z}$ and $\mathrm{d}_z \equiv \frac{\mathrm{d}}{\mathrm{d} z}$ respectively for formal partial and total derivative operators with respect to a parameter $z \in \mathbb{R}, \mathbb{C}$.

\section{Review}

We remind the reader that in the following $x \in \mathbb{Z}_+$ will be a discrete variable. In \cite{paper 01} we proved the following.

\begin{defn}
Define $L_0$ to be the operator on $\mathscr{H}$ with action
\begin{align}
	L_0v(x) = \left\{
	\begin{array}{cc}
		- (x+1)v(x+1) + (2x + 1)v(x) - x v(x-1) &,\quad x > 0 \\
		- v(1) + v(0) &,\quad x = 0 .
	\end{array} \right.
\end{align}
and domain $\mathcal{D}(L_0) := \{ v \in \mathscr{H}\ |\ || Mv || < \infty \}$, where $M$ is the multiplication operator with action $Mv(x) = xv(x)$ $\forall v \in \mathscr{T}$.
\end{defn}

\begin{unprop}[1 of \cite{paper 01}]\label{unprop01}
The operator $L_0$ has the following properties.
\begin{enumerate}
	\item $L_0$ is essentially self-adjoint.
	\item $L_{0}$ has generalized multiplicity 1.
	\item The spectrum of $L_0$ is absolutely continuous, $\sigma(L_0) = \sigma_{\mathrm{ac}}(L_0) = [0,\infty)$, and its generalized eigenfunctions are the Laguerre polynomials $\phi_{\lambda}(x) \equiv \phi^{L_0}_{\lambda}(x) = \sum_{k=0}^x \frac{(-\lambda)^k}{k!}\binom{x}{k}$ for choice of normalization $\phi_{\lambda}(0) = 1$.
\end{enumerate}
\end{unprop}

\noindent Chen, Fr\"ohlich, and Walcher determined the above properties for $L_{0}$ in \cite{CFW} via methods which are different from ours.

\begin{defn}
Let $\psi_z \equiv \psi^{L_0}_z := R^{L_{0}}_z\chi_{0}$ for all $z \in \rho(L_{0})$ be the \emph{resolvent vector}, where $\chi_0$ is the orthonormal basis vector supported at lattice site $x = 0$.
\end{defn}

\begin{unlem}[5.2 of \cite{paper 01}]
One has the representation
\begin{align}
	\psi_{z}(x) = e^{-z} \sum_{k=0}^x (-1)^k\binom{x}{k}E_{k+1}(-z),
\end{align}
where
\begin{align}
	E_p(z) := z^{p-1}\int_z^\infty \mathrm{d}t\ e^{-t} t^{-p},\qquad p\in \mathbb{C}, z\in \mathbb{C}\setminus(-\infty,0]
\end{align}
are the \emph{generalized exponential integrals} for which we take the principal branch with standard branch cut $\Sigma = (-\infty, 0]$.
\end{unlem}

\section{Results}

Consider the discrete NLS
\begin{align}\label{NLS}
	i\partial_tw = L_0w - |w|^{2\sigma}w,\quad 1 \le \sigma \in \mathbb{Z}
\end{align}
where $w : \mathbb{R}_t \times \mathbb{Z}_+ \to \mathbb{C}$. The existence of a $u: \mathbb{Z}_+ \to \mathbb{C}$ which satisfies the nonlinear finite difference equation
\begin{align}
	L_0u = \zeta u + |u|^{2\sigma}u,
\end{align}
furnishes a stationary state of the discrete NLS of the form $w(t) = e^{-i\zeta t}u$. One expects that, due to the attractive nature of the nonlinearity, a negative ``nonlinear eigenvalue'', $\zeta = -a < 0$, will allow the existence of a sharply peaked, monotonically decaying ``ground state soliton''. We will therefore exclusively look for solutions to
\begin{align}\label{solitoneq}
	L_0u = -a u + u^{p},
\end{align}
where $u: \mathbb{Z}_+ \to \mathbb{R}_{+}$, $a > 0$, and for generality $1 < p \in \mathbb{Z}$. Solutions with these characteristics are self-focusing and tend to be sharply localized. They are therefore termed solitary waves or \emph{solitons} generally.

\begin{snthm}\label{snthm04}
There exists a $\mu_* > 0$ such that for each $ \mu > \mu_*$ there exists a solution to Equation \eqref{solitoneq} with $a = \mu > 0$ and $u = \alpha_{\mu}$, which is:
\begin{enumerate}
	\item positive: $\alpha_{\mu}(x) > 0$ for all $x \in \mathbb{Z}_+$
	\item monotonically decaying: $\alpha_{\mu}(x+1) - \alpha_{\mu}(x) < 0$ for all $x \in \mathbb{Z}_+$
	\item absolutely integrable: $\alpha_{\mu} \in \ell^1$
\end{enumerate}
\end{snthm}

\begin{defn}
We define $P:= I - P_{0}$, where $P_0 := \chi_0 \otimes \chi_0^*$, and write $\widehat{v} \equiv P v$ for all $v \in \mathscr{H}$ and write $\widehat{A} \equiv P A$ for all $A \in \mathcal{B}(\mathscr{H})$.
\end{defn}

The proof of Theorem \ref{snthm04} will proceed as follows:
\begin{enumerate}
\item
	Consider Equation \eqref{solitoneq}. Split this equation into a boundary piece and a tail piece by applying $P_{0}$ and $P$ respectively.
\item
	We take $b := u(0)$ to be a fixed constant and iterate the tail piece of Equation \eqref{solitoneq} via
	\begin{align}
		u_{n+1}(a,b) = \widehat{\psi}_{-a}b^p + \widehat{R}^{L_0}_{-a}u^p_n(a,b),
	\end{align}
	$u_n(a,b) \equiv u_n(a,b,\cdot) \in \mathscr{H}$ for all $n$. We show that for large enough $a$ this iteration in $n$ converges pointwise monotonically and that $|| u_n(a,b) ||_{1} \leq s_n(a)$ is bounded as $n \nearrow \infty$, for a sequence of constants $\{s_n(a)\}_{n=0}^\infty$. We define the limit of this iteration to be $u_*(a,b) \equiv \lim_{n \nearrow \infty}u_n(a,b)$.
\item
	The construction of $u_*(a,b)$ sets $u(1) = u_*(a,b;1) \equiv q(a,b)$. We substitute this into the boundary piece of Equation \eqref{solitoneq} which then takes the form
	\begin{align}
		0 &= b^p - (a+1)b + q(a,b).
	\end{align}
	We will now take $b$ to be a variable. If the solution, $u$, is positive and monotonically decaying then one must have
	\begin{align}
		0 < u(1) < u(0) \quad \Rightarrow \quad a^{(p-1)^{-1}} < b < (a+1)^{(p-1)^{-1}} .
	\end{align}
	We show that for all $a$ sufficiently large there is a unique $b = b_*(a) \in (a^{(p-1)^{-1}},(a+1)^{(p-1)^{-1}})$ which solves the boundary equation.
\item
	We define the solution we desire, $\alpha_{\mu}$, by
	\begin{align}
		\alpha_{\mu}(x) := \left\{
		\begin{array}{cr}
			b_*(\mu) &,\ x=0\\
			u_*(\mu,b_*(\mu);x) &,\ x > 0
		\end{array}\right. .
	\end{align}
\item
	The three properties (i.e. positivity, monotonicity, $\ell^1$) of the solution will then be verified in turn.
\end{enumerate}

\noindent Typically one can arrive at the existence of a soliton with such properties via variational or rearrangement arguments. We will use much more elementary techniques which yield yet other properties due to the dependence on explicit constructions. One such result which will be of use later on is

\begin{snprop}\label{snprop02}
\begin{align}
	||\widehat{\alpha}_{\mu}||_{1} \le \mu^{-(p-1)^{-1}(p-2)} + \mathcal{O}(\mu^{-(p-1)^{-1}(2p-3)}) .
\end{align}
\end{snprop}

\noindent The spatial decay rate of $\alpha_\mu$ is generic for positive solutions vanishing at infinity.

\begin{snthm}\label{snthm2}
Consider the equation
\begin{align}
	(L_0+a)u = V(u)u,
\end{align}
where $V(\cdot): \mathbb{R} \rightarrow \mathbb{R}$ is continuous, locally bounded, and satisfies $\lim_{r \searrow 0} V(r)r = 0$. If $u$ is a solution of this equation which is positive for all $x$ and for which $\lim_{x \nearrow \infty} u(x) = 0$ then $u(x) \sim c_0 x^{-1/2}e^{- c_1 \sqrt{x}}$ as $x \nearrow \infty$ for some $0 < c_0, c_1 < \infty$, i.e. $c'_0 x^{-1/2}e^{- c'_1 \sqrt{x}} \le u(x) \le c''_0 x^{-1/2}e^{- c''_1 \sqrt{x}}$ for some $0 < c'_0, c''_0, c'_1, c''_1 < \infty$, and for each fixed $a > 0$.
\end{snthm}

\section{Existence of $\alpha_{\mu}$}

\subsection{Away from the boundary}

Consider two forms of Equation \eqref{solitoneq}:
\begin{align}
	L_0u =&\ -a u + u^p \tag{I} \\
	u =&\ R^{L_0}_{-a}u^p \tag{II} .
\end{align}

One may project the equation of form (II) to a ``tail'' piece by applying $P$:

\begin{align}
	u = R^{L_0}_{-a}(P_{0}+P)u^p \quad \Rightarrow \quad \widehat{u} = \widehat{\psi}_{-a}u^p(0) + \widehat{R}^{L_0}_{-a}\widehat{u}^p.
\end{align}

We will fix $u(0) \equiv b$ and iterate by substituting the LHS into the RHS. We will show the conditions under which this converges and estimate properties of the resulting solution.

\begin{defn}
Let $u(0) \equiv b$ be a fixed parameter which satisfies $a^{(p-1)^{-1}} < b < (a+1)^{(p-1)^{-1}}$. Let $\{ u_n(a,b) \}_{n=0}^\infty$ be a sequence of vectors in $\mathscr{H}$ defined by a fixed $u_0(a,b)$ and inductively by
\begin{align}
	u_{n+1}(a,b) = \widehat{\psi}_{-a}b^p + \widehat{R}^{L_0}_{-a}u^p_n(a,b),
\end{align}
such that $u_n(a,b) = \widehat{u}_n(a,b)$.
\end{defn}

\noindent The requirement that $a^{(p-1)^{-1}} < b < (a+1)^{(p-1)^{-1}}$ follows from $0 < u(1) < u(0)$ such that $u(0) = b$ and $u(x)$ monotonically decreasing for increasing $x$. We will be concerned with analysis of $\psi_z(x)$ for $z < 0$. For $\Re z < 0$ one has a useful presentation for the $E_p(z)$:
\begin{align}
E_p(z) := \int_1^\infty \mathrm{d}t\ e^{-z t} t^{-p}.
\end{align}

\begin{lem}
One has that $||\psi_{-a}||_{1} = a^{-1}$ for all $a > 0$.
\end{lem}

\begin{proof}
$\psi_{-a}(x) > 0$ for all $a \in \mathbb{R}_{+}$ and $x \in \mathbb{Z}_{+}$. Therefore
\begin{align}
	||\psi_{-a}||_{1} &= \sum_{x=0}^{\infty}\psi_{-a}(x) \\
		&= \sum_{x=0}^{\infty}e^{a}\sum_{k=0}^{x}(-1)^{k}\binom{x}{k}E_{k+1}(a) = \sum_{x=0}^{\infty}e^{a} \int_{1}^{\infty}\mathrm{d}t\ e^{-a t}t^{-1}(1-t^{-1})^{x} \\
		&= e^{a} \int_{1}^{\infty}\mathrm{d}t\ e^{-a t}t^{-1}\sum_{x=0}^{\infty}(1-t^{-1})^{x} = \int_{0}^{\infty}\mathrm{d}t\ e^{-a t} = a^{-1}.
\end{align}
\end{proof}

Instead of determining the precise behavior of $||u_n(a,b)||_1$ in $a$ we will construct a series of functions, $\{s_n(a)\}_{n=0}^\infty$, for which $||u_n(a,b)||_1 < s_n(a)$ for each $n$ and whose behavior in $a$ is clear.

\begin{lem}
Let $\{ s_n(a) \}_{n=0}^\infty$ be a sequence of nonnegative real numbers defined by a fixed $s_0(a)$ and inductively by
\begin{align}
	s_{n+1}(a) &= a^{-1}r(a) + a^{-1} s_n^p(a),
\end{align}
where
\begin{align}
	r(a) := (a+1)^{(p-1)^{-1}} .
\end{align}
If $\left|\left| u_j(a,b) \right|\right|_{1} \leq s_j(a)$ for some $j$ then $\left|\left| u_k(a,b) \right|\right|_{1} < s_k(a)$ for all $j < k$.
\end{lem}

\begin{proof}
Consider the known bound \cite{AandS}
\begin{align}
	(z+n)^{-1} < e^z  E_n(z) \leq (z+n-1)^{-1},\quad 0 < z \in \mathbb{R}
\end{align}
Since $\psi_{-a}(0) = e^aE_1(a)$, one then has that
\begin{align}
	\left|\left| \widehat{\psi}_{-a} \right|\right|_{1} = \left|\left| \psi_{-a} \right|\right|_{1} -\psi_{-a}(0) < a^{-1}(a+1)^{-1} .
\end{align}
Then, since $b < (a+1)^{(p-1)^{-1}}$ one has
\begin{align}
	\left|\left| \widehat{\psi}_{-a} b^p \right|\right|_{1} < a^{-1}r(a) .
\end{align}
One may then observe that
\begin{align}
	\left|\left| u_{j + 1} \right|\right|_{1} < a^{-1}r(a) + a^{-1}\left|\left| u_j^p \right|\right|_{1} \leq a^{-1}r(a) + a^{-1} s_j^p(a) = s_{j+1}(a).
\end{align}
\end{proof}

\begin{lem}
Let $g: \mathbb{R}_+^2 \to \mathbb{R}$ be specified by
\begin{align}
	g(a,s) := a^{-1}r(a) + a^{-1}s^p - s.
\end{align}
and $s_\mathrm{min}(a) := (a/p)^{(p-1)^{-1}}$. For sufficiently large $a>0$ it is the case that $g(a,s)$ has exactly two roots in $s$: $s_{-}(a)$ which satisfies $0 < s_{-}(a) < s_\mathrm{min}(a)$ and $s_+(a)$ which satisfies $s_\mathrm{min}(a) < s_+(a)$.
\end{lem}
	
\begin{proof}
One may observe that $g(a,s)$ has a global minimum at $s = s_\mathrm{min}(a)$. It is the case that
\begin{align}
	g(a,s_\mathrm{min}(a)) =&\ a^{-1}r(a) + a^{-1}(a/p)^{p(p-1)^{-1}} - a^{(p-1)^{-1}}p^{-(p-1)^{-1}}\\
	=&\ a^{-1}(a+1)^{(p-1)^{-1}} - a^{(p-1)^{-1}}p^{-(p-1)^{-1}}(1-p^{-1})\\
	\mathrm{d}_a g(a,s_\mathrm{min}(a)) &= - a^{-1}(a+1)^{(p-1)^{-1}} \left[ a^{-1} - (a+1)^{-1}(p-1)^{-1} \right]\\
		&\quad\quad - a^{-(p-2)(p-1)^{-1}}p^{-(p-1)^{-1}}(p-1)^{-1}(1-p^{-1}) < 0,\quad \forall a>0
\end{align}
since
\begin{align}
	a^{-1} - (a+1)^{-1}(p-1)^{-1} > 0,\quad \forall a>0.
\end{align}
One has that
\begin{align}
	\lim_{a \nearrow \infty} g(a,s_\mathrm{min}(a)) = -\infty
\end{align}
and
\begin{align}
	g(a,0) =&\ a^{-1}r(a) > 0,\quad \forall a>0 \\
	\lim_{a \nearrow \infty} g(a,0) =&\ \infty \\
	\lim_{s \nearrow \infty} g(a,s) =&\ \infty,\quad \forall a>0
\end{align}
By intermediate value theorem there must be at least one root. By Descartes rule of signs, $g(a,s)$ has either 0 or 2 positive roots in $s$. Therefore $g(a,s)$ has exactly 2 roots for all sufficiently large $a > 0$.
\end{proof}

\begin{defn}
Let $a_{0} > 0$ be the unique value of $a$ such that $g(a,s)$ has two distinct roots in $s$ for all $a > a_{0}$.
\end{defn}

\noindent $a_{0} > 0$ is the unique value of $a$ for which that $g(a,s)$ has one root in $s$ of multiplicity 2.
	
\begin{lem}
Let $h: \mathbb{R}_+^2 \to \mathbb{R}_+$ be defined by
\begin{align}
	h(a,s) := a^{-1}r(a) + a^{-1} s^p.
\end{align}
The map $h(a,\cdot): s \mapsto h(a,s)$ is contractive on the domain $0 \leq s < s_\mathrm{min}(a)$ for all $a > a_{0}$.
\end{lem}
	
\begin{proof}
\begin{align}
	\partial_s g(a,s)|_{s = 0} =&\ -1\\
	\partial_s g(a,s)|_{s = s_\mathrm{min}(a)} =&\ 0\\
	\partial^2_s g(a,s) >&\ 0,\quad \forall s>0\\
	\Rightarrow\quad 0 \leq \partial_s h(a,s) <&\ 1,\quad \forall s: 0 \leq s < s_\mathrm{min}(a)
\end{align}
By mean value theorem, for any $s,s_{1} \in [0, s_\mathrm{min})$ there exists an $s_{2} < \left| s-s_{1} \right|$ such that
\begin{align}
	\frac{h(a,s)-h(a,s_{1})}{s-s_{1}} = \partial_s h(a,s)|_{s = s_{2}} < 1,
\end{align}
which completes the proof.
\end{proof}

\begin{lem}
For all $a > a_{0}$ and for sufficiently small $|| u_{0}(a,b) ||_{1} \ge 0$ one has that \\ $ \lim_{n \nearrow \infty} \left|\left| u_n(a,b) \right|\right|_{1} < \infty $.
\end{lem}

\begin{proof}
Let $a > a_{0}$. Given the iteration $s_{n+1}(a) = h(a,s_n(a))$, the choice of any $s_0(a)$ which satisfies $0 \leq s_0(a) \leq s_-(a)$ gives a sequence $\{ s_n(a) \}_{n=0}^\infty$ which converges to $\lim_{n \nearrow \infty}s_{n}(a) = s_-(a) < \infty$ monotonically from below as $n \nearrow \infty$ for all $a > a_{0}$. For $0 \le || u_{0}(a,b) ||_{1} \le s_{0}(a)$ it is the case that $ \lim_{n \nearrow \infty} \left|\left| u_n(a,b) \right|\right|_{1} \leq \lim_{n \nearrow \infty}s_{n}(a) = s_{-}(a) < \infty $.
\end{proof}

\begin{lem}
Let $s_{*}(a) := \lim_{n \nearrow \infty}s_{n}(a) = s_{-}(a)$ for all $a > a_{0}$. One has that $s_{*}(a) \searrow 0$ monotonically as $a \nearrow \infty$ for all $a > a_{0}$.
\end{lem}

\begin{proof}
Let $a > a_{0}$. Consider the graph of the map $s \mapsto h(a,s) = a^{-1}r(a) + a^{-1} s^p$. It is the case that $\partial_a h(a,s) < 0$ and $h(a,s) \geq s$ for $0 \leq s \leq s_{*}(a)$, with equality for $h(a,s_{*}(a)) = s_{*}(a)$, which completes the proof.
\end{proof}

\begin{lem}
We define
\begin{align}
	\theta_1(a,b) :=&\ \widehat{\psi}_{-a} b^p, \\
	\theta_k(a,b) :=&\ \widehat{R}_{-a}^{L_0} \sum_{m_1 = 1}^p \binom{p}{m_1}\theta^{p-m_1}_1(a,b) \cdots \sum_{m_i = 1}^{m_{j-1}} \binom{m_{j-1}}{m_j}\theta^{m_j-m_{j-1}}_j(a,b) \cdots \\
		&\times \sum_{m_{k-2} = 1}^{m_{k-3}} \binom{m_{k-3}}{m_{k-2}}\theta^{m_{k-3}-m_{k-2}}_{k-2}(a,b)\theta^{m_{k-2}}_{k-1}(a,b).
\end{align}
If $u_0(a,b) = 0$ then one has
\begin{align}
	u_n(a,b) = \sum_{k=1}^n\theta_k(a,b).
\end{align}
\end{lem}

\begin{proof}
One may verify by induction:
\begin{align}
	u_1(a,b) = \theta_1(a,b)
\end{align}
\begin{align}
	u_{n+1}(a,b) &= \theta_1(a,b) + \widehat{R}^{L_0}_{-a}u^p_n(a,b) \\
		&= \theta_1(a,b) + \widehat{R}^{L_0}_{-a} \left[ \sum_{k=1}^n\theta_k(a,b) \right]^p\\
		&= \theta_1(a,b) + \widehat{R}^{L_0}_{-a} \left\{ \theta_1^p(a,b) + \sum_{m_1 = 1}^p\binom{p}{m_1} \theta_1^{p-m_1}(a,b) \left[ \sum_{k=2}^n\theta_k(a,b) \right]^{m_1} \right\}\\
		&= \theta_1(a,b) + \widehat{R}^{L_0}_{-a} \left[ \theta_1^p(a,b) + \sum_{m_1 = 1}^p\binom{p}{m_1} \theta_1^{p-m_1}(a,b) \theta^{m_1}_2(a,b) + \cdots \right.\\
			&\quad\quad + \sum_{m_1 = 1}^p \binom{p}{m_1}\theta^{p-m_1}_1(a,b) \cdots \sum_{m_j = 1}^{m_{j-1}} \binom{m_{j-1}}{m_j}\theta^{m_j-m_{j-1}}_j(a,b) \cdots \\
			&\quad\quad \times \sum_{m_{k-2} = 1}^{m_{k-3}} \binom{m_{k-3}}{m_{k-2}}\theta^{m_{k-3}-m_{k-2}}_{k-2}(a,b)\theta^{m_{k-2}}_{k-1}(a,b) \\
			&\quad\quad \left. + \sum_{m_1 = 1}^p \binom{p}{m_1}\theta^{p-m_1}_1(a,b) \cdots \sum_{m_j = 1}^{m_{j-1}} \binom{m_{j-1}}{m_j}\theta^{m_i-m_{j-1}}_j(a,b) \cdots \right. \\
			&\quad\quad \left. \times \sum_{m_{n-1} = 1}^{m_{n-2}} \binom{m_{n-2}}{m_{n-1}}\theta^{m_{n-2}-m_{n-1}}_{n-1}(a,b)\theta^{m_{n-1}}_n(a,b) \right] \\
		&= \theta_1(a,b) + \theta_2(a,b) + \cdots + \theta_j(a,b) + \cdots + \theta_{n+1}(a,b) \\
		&= \sum_{k=1}^{n+1}\theta_k(a,b) .
\end{align}
\end{proof}

\begin{lem}
For $u_0(a,b) = 0$ one has that $u_{n+1}(a,b; x) > u_n(a,b; x)$ for all $n \geq 0, x > 0$ (strict pointwise increase in $n$ for all $x > 0$) for all $a > a_{0}$.
\end{lem}
	
\begin{proof}
Since $\theta_n(a,b; x) > 0$, $\forall x>0$, and $u_{n+1}(a,b) - u_n(a,b) = \theta_{n+1}(a,b)$ it must be the case that the sequence $\{ u_n(a,b; x) \}_{n=0}^\infty$ is strictly increasing in $n$ for all $x>0$ and for all $a > a_{0}$.
\end{proof}
	
\begin{lem}
For $u_0(a,b) = 0$ and all $a>a_{0}$ one has that $\lim_{n \nearrow \infty} u_n(a,b)  = \sum_{n=1}^\infty \theta_n(a,b) \in \ell^1$.
\end{lem}
	
\begin{proof}
Let $a>a_{0}$. One may observe that that $\lim_{n \nearrow \infty} \left|\left| u_n(a,b) \right|\right|_{1} \le s_*(a) < \infty \Rightarrow \lim_{n \nearrow \infty} u_n(a,b; x) < s_*(a) < \infty$ for all $x \in \mathbb{Z}_{+}$. Then by monotonic increase of $u_n(a,b; x)$ in $n$ for all $x > 0$, it follows that $u_n(a,b; x) = \sum_{k=1}^n\theta_k(a,b; x)$ exists for each $n,x$ and uniquely determines $\lim_{n \nearrow \infty} u_n(a,b)$.
\end{proof}

\begin{lem}
Let $u_{*}(a,b) := \sum_{n=1}^\infty \theta_n(a,b) = \lim_{n \nearrow \infty} u_n(a,b)$ for all $a > a_{0}$. One has that $u_*(a,b; x)$ is a monotonically increasing function in $b \geq 0$ for all $x > 0$ and for all $a > a_{0}$.
\end{lem}
	
\begin{proof}
$u_*(a,b; x) := \sum_{n=1}^\infty \theta_n(a,b; x)$ can be represented as a power series in $b$ for all $x > 0$ with only positive powers and positive coefficients.
\end{proof}

\subsection{At the boundary}

We now consider the equation of form (I). We can project it onto a ``boundary'' piece by applying $P_{0}$:
\begin{align}
	0 = (P_{0} + P)(-L_0u - au + u^p) \quad \Rightarrow \quad 0 = u^p(0) - (a+1)u(0) + u(1)
\end{align}

Given $u_{*}(a,b) = \sum_{k=1}^\infty \theta_k(a,b)$ we will substitute $b = u(0)$ and $q(a,b) \equiv u_*(a,b; 1) = u(1)$ in the boundary equation and thereby consider
\begin{align}
	0 = b^p - (a+1)b + q(a,b)
\end{align}
on the interval $a^{(p-1)^{-1}} \leq b \leq (a+1)^{(p-1)^{-1}}$.

\begin{defn}
\begin{align}
	&b_-(a) := a^{(p-1)^{-1}},\quad b_+(a) := (a+1)^{(p-1)^{-1}} \\
	&\Sigma_a := \left\{ b \in \mathbb{R} : b_-(a) \leq b \leq b_+(a) \right\} \\
	&q_-(a,b) := 0,\quad q(a,b) := u_*(a,b; 1),\quad q_+(a,b) := b \\
	&f(a,b,q) := b^p - (a+1)b + q \\
	&f_-(a,b) := f(a,b, q_-(a,b)),\quad f_*(a,b) := f(a,b, q(a,b)), \\
		&\quad f_+(a,b) := f(a,b, q_+(a,b))
\end{align}
\end{defn}
	
\begin{figure}
	\centering
	\includegraphics[height=70mm]{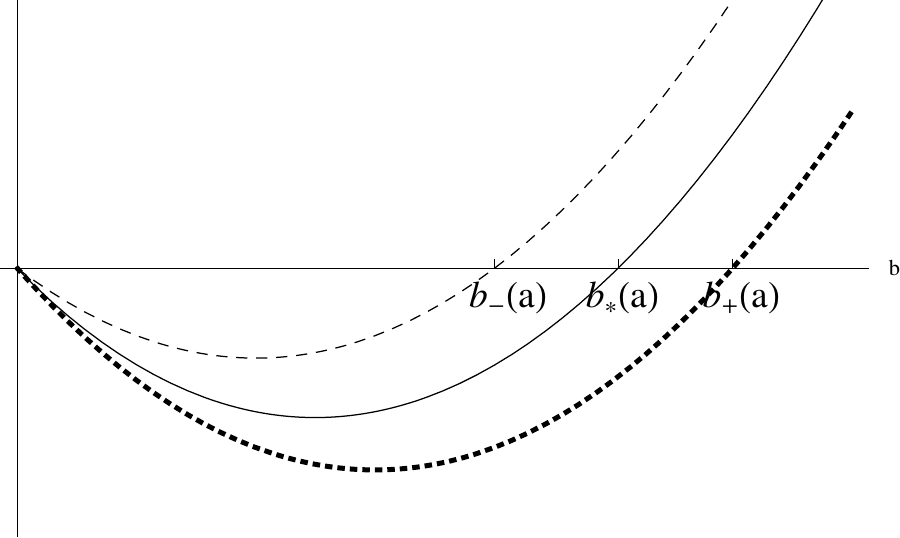}
	\caption{The dashed line is the graph of $f_+(a,b)$, the solid is of $f(a,b)$, and the dotted is of $f_-(a,b)$.}
\end{figure}

\begin{lem}
$q(a,b) > q_-(a,b)$ for all $ b \in \Sigma_a$ and for all $a>a_{0}$.
\end{lem}
	
\begin{proof}
Let $a>a_{0}$.
\begin{align}
	q(a,b) &= \sum_{n = 1}^\infty \theta_n(a,b; 1) >  \theta_1(a,b; 1) = \widehat{\psi}_{-a}(1) b^p \\
		&= e^a\left[ E_1(a) - E_2(a) \right] b^p > 0 = q_-(a,b)
\end{align}
since $E_2(a) < E_1(a),\ \forall a>0$ \cite{AandS}.
\end{proof}
	
\begin{lem}
There exists an $a_{1} \geq a_{0}$ such that $q(a,b) < q_+(a,b) $ for all $b \in \Sigma_a$ and for all $a > a_{1}$.
\end{lem}
	
\begin{proof}
Let $a > a_{0}$. It is the case that $\lim_{a \nearrow \infty} s_*(a) = 0$ monotonically and $\lim_{a \rightarrow \infty} b_-(a) = \infty$ monotonically so there exists an $a_{1} \ge a_{0}$ such that $s_*(a) < b_-(0),\ \forall a > a_{1}$. Clearly $q(a,b) < s_*(a)$. Then $q(a,b) < s_*(a) < b_-(a) \leq b = q_+(a,b),\ \forall a > a_{1}$ and $\forall b \in \Sigma_a$.
\end{proof}

\begin{lem}
Let $a_{1} > 0$ be the smallest value for which $a_{1} \ge a_{0}$ and for which $q(a,b) < q_+(a,b) $ for all $b \in \Sigma_a$ and for all $a > a_{1}$. Define $a_{2} := (p-1)^{-1}$, $a_{3} := \max\{a_{1},a_{2}\}$. One has that:
\begin{itemize}
	\item $f_-(a,b)$ is negative and monotonically increasing for all $b \in \Sigma_a \setminus b_+(a)$ and for all $a > a_{3}$.
	\item $f_+(a,b)$ is positive and monotonically increasing in $b$ for all $b \in \Sigma_a \setminus b_-(a)$ and for all $a > a_{3}$.
\end{itemize}
\end{lem}
	
\begin{proof}
One may observe that $f_-(a,b) < f_*(a,b) < f_+(a,b)$ for all $a>a_{1}$ and all $b \in \Sigma_a$. $f_-(a,b_+(a)) = 0$ and $f_+(a,b_-(a))=0$ for all $a \geq 0$.

\underline{For $f_-(a,b)$:}
Let $a > a_{3}$. It follows that
\begin{align}
	\partial_b f_-(a,b) &= p b^{p-1} - (a+1) > p b_-^{p-1}(a) - (a+1) = a(p-1) - 1.
\end{align}
Then $\partial_bf_-(a,b) > 0,\ \forall b \in \Sigma_a$ and for all $a > a_{3}$. Since $f_-(a,b_+(a)) = 0$ one has that $f_-(a,b) < 0$ for all $b \in \Sigma_a \setminus b_+(a)$.

\underline{For $f_+(a,b)$:}
Let $a > a_{3}$.
\begin{align}
	\partial_bf_+(a,b) &= p b^{p-1} - a \geq p b_-^{p-1}(a) - a = a(p-1) > 0,\ \forall a > 0.
\end{align}
Then since $f_+(a,b_-(a)) = 0$ one finds that $f_+(a,b) > 0$ $\forall b \in \Sigma_a \setminus b_-(a)$.
\end{proof}
	
\begin{lem}
$f_*(a,b)$ is monotonically increasing in $b$ and has exactly one root in $b$ for all $b \in \Sigma_a$ and for all $a > a_{3}$.
\end{lem}
	
\begin{proof}
Let $a > a_{3}$. $f_-(a,b)$ is monotonically increasing in $b$ for all $b \in \Sigma_a$. $q(a,b)$ is monotonically increasing in $b$ for all $b > 0$. $f_*(a,b) = f_-(a,b) + q(a,b)$ so $f_*(a,b)$ must be monotonically increasing in $b$ for all $b \in \Sigma_a$.

Since $f_*(a,b_-(a)) < 0 = f_+(b_-(a)) = f_-(a,b_+(a)) < f_*(a,b_+(a))$, by intermediate value theorem there must be at least one $b = b_*(a)$ for which $f_*(a,b) = 0$ and since $f_*(a,b)$ is monotonically increasing on this interval, there must be exactly one such $b = b_*(a)$.
\end{proof}

\begin{defn}\label{part01}
Let $b_*(a)$ be the unique value of $b \in \Sigma_a$ for which $f_*(a,b)=0$ for all $a > a_{3}$ and define $\alpha_{a} \in \mathscr{H}$ via
\begin{align}
	\alpha_{a}(x) := \left\{
	\begin{array}{cr}
		b_*(a) &,\ x=0\\
		u_*(a,b_*(a)) &,\ x > 0
	\end{array}\right. ,
\end{align}
for all $a > a_{3}$.
\end{defn}

\section{Properties of $\alpha_{\mu}$}

\subsection{Monotonicity}

Consider
\begin{align}\label{monotoneeq}
	-L_0u = V(u),
\end{align}
where $V(\cdot): \mathbb{R} \rightarrow \mathbb{R}$ satisfies $V(r) = 0$, $V(r_0) = 0$, and $V(r) > 0$ for all $r \in (0, r_0)$ where $r_0 > 0$. It was shown in \cite{DJN 1} that for any solution of this equation for which $\lim_{x \nearrow \infty} u(x) = 0$, there exists an $x_*$ such that $ u(x) $ is a monotonically decaying for increasing $x \geq x_*$. We will apply their argument, section 3 of \cite{DJN 1}, to show that for sufficiently large $a$ one must have that $\alpha_{a}(x)$ is monotonically decaying in $x$ for all $x \in \mathbb{Z}_+$. For our case we have that $V(r) = ar -r^p$, $r \in \mathbb{R}$, which satisfies the desired criteria.

Equation \eqref{monotoneeq} may be summed to give an alternative finite difference equation. One has the two forms:
\begin{align}
	(x+1)[u(x+1) - u(x)] - x[u(x) - u(x-1)] &= V(u(x)),\\
	u(x+1) - u(x) &= (x+1)^{-1} \sum_{y = 0}^x V(u_y).
\end{align}
The argument proceeds as follows.	
\begin{lem}
There exists an $a_{4} \ge a_{3}$ such that $s_*(a) < b_-(a)$ for all $a > a_{4}$.
\end{lem}
	
\begin{proof}
Let $a > a_{3}$. It is the case that $s_*(a) \searrow 0$ monotonically as $a \nearrow \infty$. It is clear that $b_-(a) \nearrow \infty$ monotonically as $a \nearrow \infty$.
\end{proof}
	
\begin{lem}
	Let $u \in \mathscr{H}$ solve Equation \eqref{monotoneeq}. Let $c_-$ and $c_+$ be constants which satisfy $0 < c_- < u(0)$ and $0 < || P u || < c_+ < \infty$. If additionally $c_+ < c_-$ then $u(x)$ is monotonically decreasing in $x$ for all $x \in \mathbb{Z}_+$.
\end{lem}
	
\begin{proof}
Let $u$ solve Equation \eqref{monotoneeq}. Let $c_-$ and $c_+$ be constants which satisfy $0 < c_- < u(0)$ and $0 < || P u || < c_+ < \infty$. Furthermore let  $c_+ < c_-$. One has that $u(x) < u(0)$ for all $x > 0$ and in particular $u(1) < u(0)$. Assume that there exists an $0 < x_0 \in \mathbb{Z}_+$ such that $u(x_0 + 1) - u(x_0) = 0$ for $u(x_0),\ u(x_0+1) \in (0, c_-)$. One then has $u(x_0+2) - u(x_0 + 1) > 0$ since $(x_0+2)[u(x_0+2) - u(x_0 + 1)] = V(u(x_0 + 1)) > 0$. Then assume generally that there exists an $0 < x_1 \in \mathbb{Z}_+$ such that $u(x_1 + 1) - u(x_1) > 0$. This gives that $\sum_{y = 0}^{x_1} V(u(y)) > 0$ since $u(x_1+1) - u(x_1) = (x_1+1)^{-1} \sum_{y = 0}^{x_1} V(u(y)) > 0$. One then has that
\begin{align}
	&u(x_1+2) - u(x_1 + 1) = (x_1+2)^{-1} \sum_{y = 0}^{x_1 + 1} V(u(y)) \\
	&\qquad = (x_1+2)^{-1}V(u(x_1 + 1)) + (x_1+2)^{-1} \sum_{y = 0}^{x_1} V(u(y)) > 0
\end{align}
since $V(u(x_1 + 1))$ and $\sum_{y = 0}^{x_1} V(u(y))$ are both positive. If $u(x_1 + 2) \geq c_-$ then one has a contradiction. On the other hand if $u(x_1 + 2) < c_-$ then one may repeat the above process with $x_1$ replaced with $x_1 + 1$. Therefore if there exists an $x_1 > 0$ such that $u(x_1+ 1) - u(x_1) > 0$, the $u(x)$ for the subsequent $x > x_1 + 1$ will continue to rise at least until $u(x_3) \geq c_-$, which is the point greater than which $V(x)$ remains negative, for some $x_3 > x_2+ 1$. One therefore has a contradiction if $u(x)$ fails to be monotonically decreasing as $x \nearrow \infty$ for all $x \in \mathbb{Z}_+$.
\end{proof}

\begin{lem}\label{part02}
Let $\mu_{*} > 0$ be the smallest value such that $\mu_{*} \ge a_{3}$ and such that $s_*(\mu) < b_-(\mu)$ for all $\mu > \mu_{*}$. One has that $\alpha_{\mu}(x)$ decreases monotonically as $x \nearrow \infty$ for all $x \in \mathbb{Z}_+$ and for all $\mu > \mu_{*}$
\end{lem}

\begin{proof}
Consider Lemma 4.1 and let $u = \alpha_{\mu}$, $c_- = b_-(\mu)$, and $c_+ = s_*(\mu)$ for all $\mu > \mu_{*}$.
\end{proof}

\begin{proof}[Proof of Theorem \ref{snthm04}]
Existence and Property (1) are given by Definition \ref{part01} and arguments on which the definition depends. Property (2) is given by Lemma \ref{part02}. Property (3) is given by the fact that the existence and construction of $\alpha_\mu$ requires $\ell^1$ boundedness.
\end{proof}

\begin{proof}[Proof of Proposition \ref{snprop02}]
One may use the convergence for the iteration of the $s_{n}(\mu)$ functions for large values of $\mu$ to directly verify:
\begin{align}
	||\widehat{\alpha}_{\mu}||_{1} \le s(\mu) = \mu^{-1}(\mu+1)^{(p-1)^{-1}}+\mu s^{p}(\mu) = \mu^{-(p-1)^{-1}(p-2)} + \mathcal{O}(\mu^{-(p-1)^{-1}(2p-3)}) .
\end{align}
\end{proof}

\subsection{Asymptotic behavior}

In \cite{CFW} the asymptotic behavior of solutions of the finite difference nonsingular Sturm-Liouville problem $P (L_0 - \lambda) P u = 0$ were studied with various boundary conditions. One solution, $\phi_\lambda$, has the well-known asymptotic behavior
\begin{align}
	\phi_{\lambda}(x) \sim e^{\lambda/2}J_0(2\sqrt{\lambda x}) \quad \text{as} \quad x \nearrow \infty,
\end{align}
where $J_0(z)$ is the Bessel function of the first kind of degree $0$. They studied a particular solution, which they call $\Psi_{\lambda}$, which is a linear combination of $\phi_\lambda$ and $\psi_{\lambda}$. It satisfies the boundary conditions $\Psi_{\lambda}(0) = 0$ and $\Psi_{\lambda}(1) = 1$. $\Psi_{\lambda}$ was shown to have the asymptotic behavior
\begin{align}
	\Psi_{\lambda}(x) &\sim \pi e^{-\lambda/2}Y_0(2\sqrt{\lambda x}) + e^{\lambda/2}\mathcal{P}E_1(-\lambda)J_0(2\sqrt{\lambda x}) \quad \text{as} \quad x \nearrow \infty,
\end{align}
where $Y_0(z)$ is the Bessel function of the second kind of degree 0.
\begin{lem}
One has the asymptotic decay rate
\begin{align}
	\psi_{-a}(x) \sim 2e^{a/2}K_0(2\sqrt{ax}) \sim e^{a/2}\pi^{1/2}(ax)^{-1/4}e^{-2\sqrt{ax}} \quad \text{as} \quad x \nearrow \infty,
\end{align}
where $K_0(z)$ is the modified Bessel function of the second kind of degree 0.
\end{lem}
\begin{proof}
One can determine the asymptotic behavior of $\psi_{-a}(x)$ as $x \nearrow \infty$ by finding the appropriate linear combination of Bessel functions such that upon analytic continuation, the real part is of the appropriate linear combination of the asymptotic forms of $\phi_\lambda(x)$ and $\Psi_{\lambda}(x)$ is monotonically decreasing as $x \nearrow \infty$. This combination must be monotonically decaying for $\lambda = -a < 0$ one can straightforwardly determine that the asymptotic behavior must be of the desired form.
\end{proof}

\begin{lem}
$\phi_{-a}(x), \psi_{-a}(x) > 0$ for all $x \in \mathbb{Z}_{+}$ and $a > 0$. $\phi_{-a}(x)$ is monotonically increasing and $\psi_{-a}(x)$ is monotonically decreasing in increasing $x$ for all $a > 0$.
\end{lem}

\begin{proof}
One can see that $\phi_{-a}(x)$ is positive and monotonically increasing by inspection of $\phi_{-a}(x) = \sum_{k=0}^x \frac{a^k}{k!}\binom{x}{k}$. One can observe that $\psi_{-a}(x)$ is positive for all $x,a$
\begin{align}
	\psi_{-a}(x) = e^a \sum_{k=0}^x(-1)^k\binom{x}{k}\mathrm{E}_{k+1}(a) = e^a \int_1^\infty\mathrm{d}t\ e^{-at}t^{-1}\left(1-t^{-1}\right)^x>0
\end{align}
as well as monotonically decreasing
\begin{align}
	\psi_{-a}(x+1) -\psi_{-a}(x) = - e^a \int_1^\infty\mathrm{d}t\ e^{-at}t^{-2}\left(1-t^{-1}\right)^x<0 .
\end{align}
\end{proof}

\begin{lem}
The resolvent $R^{L_0}_z$ has the Sturm-Liouville (SL) representation
\begin{align}
	R^{L_0}_z(x_{1},x_{2}) = \left\{
	\begin{array}{cr}
		\phi_z(x_{1}) \psi_z(x_{2}),\ \  x_{1}\le x_{2}\\
		\phi_z(x_{2}) \psi_z(x_{1}),\ \  x_{1}\ge x_{2}
	\end{array}\right. .
\end{align}
\end{lem}

\begin{proof}
The operator $L_0$ on $\mathscr{H}$ is a singular, second order, finite difference Sturm-Liouville operator. This is made manifest when put into SL form, $L_0 = D_+MD_-$, where $D_{+}, D_{-}$ are the respectively the usual forward and backward finite difference operators
\begin{align}
	&D_+v(x)=v(x+1)-v(x)\\
	&D_-v(x)=\left\{
	\begin{array}{cr}
		v(x)-v(x-1) &,\ x>0\\
		v(x) &,\ x=0,
	\end{array}\right.
\end{align}
and $M$ is the lattice index multiplication operator $Mv(x) = xv(x)$ for all $v \in \mathscr{T}$.
	
$L_{0}$ is singular at the boundary point $x=0$. When its domain and range are restricted to functions only on lattice points for $x > 0$ it is the case that $L_{0}$ is a non-singular operator. This restricted operator, $P L_{0} P$, is second order and nonsingular therefore $P L_0 P u = z u$, $u \in \mathscr{T}$, admits two linearly independent solutions which satisfy linearly independent boundary conditions.

The finite difference Wronskian, also known as the Cassoratian, of two vectors $u,v\in\mathscr{T}$ is given by (see e.g. \cite{Casoratian})
\begin{align}
	W[u,v](x) = u(x)v(x+1) - u(x+1)v(x) .
\end{align}
A Jacobi operator $A \in \mathcal{L}(\mathscr{T})$ can be brought into the form
\begin{align}
	Av(x) = \eta(x)v(x+1) + \omega(x)v(x) +\eta(x-1)v(x-1),\quad \eta(x),\omega(x) \in \mathbb{R}.
\end{align}
As is specified by the finite difference Sturm-Liouville theory, if $A$ is a Sturm-Liouville operator, $u_z,v_z$ are linearly independent solutions of $Au=zu$, and $\eta(x)W[u_z,v_z](x) = 1$ $\forall x$ then
\begin{align}
	R^A_z(x_{1},x_{2}) = \left\{
	\begin{array}{cr}
		u_z(x_{1}) v_z(x_{2}),\ \  x_{1}\le x_{2}\\
		u_z(x_{2}) v_z(x_{1}),\ \  x_{1}\ge x_{2}
	\end{array}\right. .
\end{align}
This construction follows for $A = P L_{0} P$, $u_{z} = P\phi_{z}$, $v_{z} = P\psi_{z}$. Since $\phi^{L_{0}}_{\lambda}(x)$ is a polynomial of degree $x$ in $\lambda$, one has that the analytic continuation $\phi^{L_{0}}_{z} \in \mathscr{T}$, $z \in \mathbb{C}$, exists. Since $L_{0}$ is singular at $x = 0$ one cannot adjust boundary conditions any further than fixing a scale factor for spectral solutions. One can simulate a second boundary condition with the introduction of either a linear perturbation or of an inhomogeneous source supported at the singular point. This is to say that one can respectively consider the equations
\begin{align}
	(L_{0} - qP_{0})u = zu \quad \text{or} \quad L_{0}u = zu + q\chi_{0}
\end{align}
where $q \in \mathbb{C}$ is a parameter, the tuning of which simulates the tuning of a second boundary condition. By taking the latter form with $q=1$ one may arrive at $\psi_{z}$ for the second solution. It therefore must be the case that $P R^{L_{0}}_{z}P = P R^{P L_{0} P}_{z}P$ with the above prescription and $R^{L_{0}}_{z} = R^{P L_{0} P}_{z}$ may be checked for $x = 0$ directly at boundary values.
\end{proof}
\noindent One may observe that one has $R^{L_0}_{-a}(x_{1},x_{2}) > 0$ for all $a>0$ and $0\leq x_{1},x_{2} \in \mathbb{Z}$.

\begin{proof}[Proof of Theorem \ref{snthm2}]

Consider $a > 0$ a fixed constant. Let $P_{\le x_*} := \sum_{x = 0}^{x_*} P_x$ and $P_{> x_*} := I - P_{\le x_*}$ for some $x_* \in \mathbb{Z}_+$. Let $q := ||P_{> x_*}V(u)  ||_\mathrm{op} = ||P_{> x_*}V(u)  ||_\infty$. Furthermore let $x_*$ satisfy $0 < q < a$ for all $a > 0$. One may find
\begin{align}
	(L_0+a)u &= V(u)u \\
	\Rightarrow \quad u &= [L_0 + a - P_{> x_*(a)}V(u)]^{-1} P_{\leq x_*(a)}u \\
		&= R^{L_0}_{-a}[1 - P_{> x_*}V(u) R^{L_0}_{-a}]^{-1} P_{\leq x_*}u \\
		&= R^{L_0}_{-a} \sum_{n = 0}^\infty [P_{> x_*}V(u) R^{L_0}_{-a}]^n P_{\leq x_*}u
\end{align}
where the sum converges absolutely. One has from the Sturm-Liouville form of the resolvent
\begin{align}
	&R^{L_0}_{-a}\chi_x \leq \phi_{-a}(x)R^{L_0}_{-a}\chi_0 = \phi_{-a}(x)\psi_{-a},
\end{align}
for all $a > 0$ and all $x \in \mathbb{Z}_+$. One then has, for all $a > 0$ and all $x \in \mathbb{Z}_+$
\begin{align}
	u(x) &\leq R^{L_0}_{-a} \sum_{n = 0}^\infty \left[q(u,x_*(a)) R^{L_0}_{-a}\right]^n P_{\leq x_*(a)}u(x) \\
		&\leq \left[ R^{L_0}_{-a + q} \sum_{y=0}^{x_*} u(y)\chi_y \right](x) \\
		&\leq \sum_{y=0}^{x_*} u(y)\phi_{-a + q}(y)\psi_{-a + q}(x)
\end{align}
One may therefore conclude
\begin{align}
	u(x) \leq c \psi_{-a + q}(x),
\end{align}
where $c = \sum_{y=0}^{x_*} u(y)\phi_{-a + q}(y) < \infty$. One has that $\psi_{-a + q}(x) \sim c' x^{-1/2}e^{-2\sqrt{(a - q)x}}$ as $x \nearrow \infty$ with $a$ fixed, where $c'$ is a constant that depends on $a - q$ alone. By inspection the appropriate constants $0 < c_0, c_1 < \infty$ may be found for each fixed $a > 0$.

\end{proof}

\thanks{We thank Marius Beceanu for helpful discussions. This work was supported in part by NSF grant DMS 1201394}


\begin{thebibliography}{1}


\bibitem{paper 01} A.J. Krueger and A. Soffer. Dynamics of Noncommutative Solitons I: Spectral Theory and Dispersive Estimates. Preprint.


\bibitem{paper 03} A.J. Krueger and A. Soffer. Dynamics  of Noncommutative Solitons II: Spectral Theory, Dispersive Estimates and Stability. Preprint.




\bibitem{AandS} M. Abramowitz, I. Stegun, Handbook of Mathematical Functions. http://people.math.sfu.ca/$\sim$cbm/aands/page\_229.htm, July 20 2010.

\bibitem{fuzzy physics} S. Baez, A. P. Balachandran, S. Vaidya, B. Ydri. Monopoles and Solitons in Fuzzy Physics. Commun.Math.Phys. 208 (2000) 787-798.

\bibitem{BaFl1} F. Bayen, M. Flato, C. Fronsdal, A. Lichnerowicz, and D. Sternheimer.  Deformation Theory and Quantization. I. Deformations of Symplectic Structures. Annals of Physics 111, 61-110 (1978).

\bibitem{BaFl2} F. Bayen, M. Flato, C. Fronsdal, A. Lichnerowicz, and D. Sternheimer.  Deformation Theory and Quantization. II. Physical Applications. Annals of Physics 111, 111-151 (1978).

\bibitem{important results 1} V.S. Buslaev, G.S. Perelman. On the stability of solitary waves for nonlinear Schr\"odinger equations. Amer. Math. Soc. Transl. 164 (1995), 75Ð98.

\bibitem{important results 2} V.S. Buslaev, C. Sulem. On asymptotic stability of solitary waves for nonlinear Schr\"odinger equations. Ann. Inst. H. Poincar\'e Anal. Non Lineare 20 (2003), 419Ð475.

\bibitem{CFW} T. Chen, J. Froehlich, J. Walcher. The Decay of Unstable Noncommutative Solitons. Commun.Math.Phys. 237 (2003) 243-269.

\bibitem{Segev review} Z. Chen, M. Segev, D.N. Christodoulides. Optical spatial solitons: historical overview and recent advances. Rep. Prog. Phys. 75 (2012) 086401 (21pp).


\bibitem{CucTar} S. Cuccagna, M. Tarulli. On Asymptotic Stability of Standing Waves of Discrete Schr\"odinger Equation in $\mathbb{Z}^*$. SIAM J. Math. Anal. Vol. 41, No. 3, pp. 861-885, 2009.

\bibitem{Derrick} G. H. Derrick. Comments on Nonlinear Wave Equations as Models for Elementary Particles. J. Math. Phys. 5, 1252 (1964); doi: 10.1063/1.1704233.

\bibitem{noncommutative scattering} B. Durhuus, V. Gayral. The Scattering Problem for a Noncommutative Nonlinear Schr\"odinger Equation. SIGMA 6 (2010), 046, 17 pages.

\bibitem{DJN 1} B. Durhuus, T. Jonsson, R. Nest. Noncommutative scalar solitons: existence and nonexistence.  Phys. Lett. B 500 (2001), no. 3-4, 320-325.

\bibitem{DJN 2} B. Durhuus, T. Jonsson, R. Nest. The Existence and Stability of Noncommutative Scalar Solitons. Comm. Math. Phys. 233 (2003), no. 1, 49-78.

\bibitem{1D lattice decay estimates} I. Egorova, E. Kopylova, G. Teschl. Dispersion Estimates for One-dimensional Discrete Schršdinger and Wave Equations. arXiv:1403.7803v1

\bibitem{important results 7} H.S. Eisenberg, Y. Silberberg, R. Morandotti, A.R. Boyd, J.S. Aitchison. Discrete spatial optical solitons in waveguide arrays. Phys. Rev. Lett. 81, 3383-3386 (1998).

\bibitem{important results 3} Z. Gang, I.M. Sigal. Asymptotic stability of nonlinear Schr\"odinger equations with potential. Rev. Math. Phys. 17 (2005), 1143-1207.

\bibitem{GMS} R. Gopakumar, S. Minwalla, A. Strominger. Noncommutative Solitons. JHEP 0005: 020, 2000.

\bibitem{JenKat} A. Jensen, T. Kato. Spectral Properties of Schr\"odinger Operators and Time-Decay of the Wave Functions. Duke Math. J. Vol. 46, No. 3. Sept. 1979.


\bibitem{Casoratian} W. Kelley, A. C. Peterson. Difference Equations. 2nd ed. Academic Press (2001).

\bibitem{important results 6} P.G. Kevrekidis, D.E. Pelinovsky, and A. Stefanov. Asymptotic stability of small solitons in the discrete nonlinear Schr\"odinger equation in one dimension. (2008). Mathematics and Statistics Department Faculty Publication Series. Paper 1143.


\bibitem{important results 4} A. Komech, E. Kopylova, M. Kunze. Dispersive estimates for 1D discrete Schr\"odinger and Klein-Gordon equations, Appl. Anal. 85 (2006), 1487Ð1508.

\bibitem{2D} E.A. Kopylova, A.I. Komech. Long time decay for 2D Klein-Gordon equation. Journal of Functional Analysis 259 (2010) 477Ð502.

\bibitem{Murata} M. Murata. Asymptotic Expansions in Time for Solutions of Schr\"odinger-Type Equations. J. of Fun. Anal. 49, 10-56 (1982).

\bibitem{NC soliton survey} O. Lechtenfeld. Noncommutative Solitons. AIPConf.Proc.977:37-51,2008.


\bibitem{important results 5} F. Palmero, R. Carretero-Gonz\'alez, J. Cuevas, P.G. Kevrekidis, W. Kr\'olikowski. Solitons in one-dimensional nonlinear Schr\"odinger lattices with a local inhomogeneity. Phys. Rev. E 77, 036614 (2008).







\bibitem{Avy NLS} A. Soffer, M.I. Weinstein. Multichannel Nonlinear Scattering for Nonintegrable Equations. Commun. Math. Phys. 133,119-146 (1990)

\bibitem{Avy NLKG} A. Soffer, M.I. Weinstein. Resonances, radiation damping and instability in Hamiltonian nonlinear wave equations, Invent. Math. 136 (1999), no. 1, 9-74.







\bibitem{En} Digital Library of Mathematical Functions. http://dlmf.nist.gov/8.19. Feb. 29 2012.

\bibitem{generating function} Digital Library of Mathematical Functions. http://dlmf.nist.gov/18.12. Sept. 21 2014.





\end{thebibliography}
\end{document}